\documentclass[12pt]{amsart}

\usepackage{textcomp}
\usepackage{amsthm}
\usepackage{amssymb}
\usepackage{amsmath}
\numberwithin{equation}{section}
\newcommand{\bea}{\begin{eqnarray}}
\newcommand{\eea}{\end{eqnarray}}
\newcommand{\be}{\begin{eqnarray*}}
\newcommand{\ee}{\end{eqnarray*}}
\newtheorem{theorem}{Theorem}[section]

\newtheorem{definition}{Definition}[section]
\newtheorem{proposition}{Proposition}[section]
\newtheorem{example}{Example}[section]
\newtheorem{algorithm}{Algorithm}[section]

\begin{document}
\title[Dynamics of Boolean Networks]{Dynamics of Boolean Networks}
\author[Yi Ming Zou]{Yi Ming Zou}
\address{Department of Mathematical Sciences, University of Wisconsin, Milwaukee, WI 53201, USA} \email{ymzou@uwm.edu}
\maketitle
\date{}

\begin{abstract}
Boolean networks are special types of finite state time-discrete dynamical systems. A Boolean network can be described by a function from an $n$-dimensional vector space over the field of two elements to itself. A fundamental problem in studying these dynamical systems is to link their long term behaviors to the structures of the functions that define them. In this paper, a method for deriving a Boolean network's dynamical information via its disjunctive normal form is explained. For a given Boolean network, a matrix with entries $0$ and $1$ is associated with the polynomial function that represents the network, then the information on the fixed points and the limit cycles is derived by analyzing the matrix. The described method provides an algorithm for the determination of the fixed points from the polynomial expression of a Boolean network. The method can also be used to construct Boolean networks with prescribed limit cycles and fixed points. Examples are provided to explain the algorithm. 
\end{abstract}
\section{Introduction}
\par
\medskip
Boolean networks have a wide range of applications, such as in computer science, engineering, computational biology, physics, and psychology \cite{Bar3, Davi1, Dub2, Ilic, Kauf1, Mol, Rege, Riel}. To facilitate the study of Boolean networks, in particular, to use the fast developing tools in computational algebra, one uses a polynomial function $f: F_2^n\longrightarrow F_2^n$ to represent a Boolean network, where $F_2=\{0,1\}$ is the field of two elements and $F_2^n$ is the $n$-dimensional vector space over $F_2$. The elements of $F_2^n$ are called {\it states}. The dynamics of the system is obtained by iterating the function $f$. The {\it state space} $S(f)$ of $f$ can be represented by a directed graph defined as follows. The vertices of $S(f)$ are the elements of $F_2^n$. There is a directed edge $a\longrightarrow b$ in $S(f)$ if $f(a) = b$. A directed edge from a vertex to itself is admissible and is called a {\it loop}. Thus, $S(f)$ encodes all state transitions of $f$, and has the property that every vertex has out-degree exactly equal to 1. Each connected graph component of $S(f)$ consists of a directed cycle called a {\it limit cycle}, with a directed tree attached to each vertex in the cycle, consisting of the {\it transients} \cite{Col1}.
\par
\medskip
When the state space is small, i.e., $n$ is small, the enumeration of the state space is an effective and intuitive way to analyze a Boolean network. If the state spaces are large, such as those appear in the modeling of complex biological systems, enumerations of the state spaces are impractical, and it is desirable to find ways to link the structure of the Boolean polynomial function $f$ to its dynamics. The study of Boolean functions has a long history and can be dated back to the middle of 19th century when Boole published his books \cite{Boo1, Boo2}. One can find an extensive bibliography in \cite{Rud1}. However, the investigation of the linkage between the structure of a Boolean function and its dynamics, in particular the development of efficient algorithms for handling substantial computations related to real applications, seem to be quite recent \cite{Bric, Col1, Gro2, He, Jar1, Jar2, Jar4, Laub, Tamu1, Zhan}, and so far, few algorithms are available. In general, the problems arise in deriving the information on a Boolean network's long term behavior from the structure of the system's defining function are believed to be NP-hard \cite{Jus1, Tamu1, Zhan}, and we are forced to consider either algorithms that are not necessary polynomial time nature in the dimension (number of parameters) of the space but useful in practice, or to restrict ourselves to some special classes of Boolean networks that we can develop effective approaches for. In \cite{Col1, Jar4} the dependency graphs (network topology) for monomial Boolean networks and conjunctive (or disjunctive) Boolean networks are analyzed to derive the information on the dynamics of these systems. However, to study more general Boolean networks based on the dependency graph seems to be difficult, since the dependency graphs carry insufficient information when functions consist of more than one mix type terms: there are too many ways a function can depend on the same set of inputs. We note that the reverse problem, namely to construct networks for a given dependency graph, was considered in \cite{Jar1, Laub} by using computational algebra tools. The networks considered in \cite{Jar1, Laub} are more general and include Boolean networks as special cases.   
\par
\medskip
In this paper, we describe a method on deriving the information about the fixed points and the limit cycles of a Boolean network from its polynomial function. Our approach is based on the disjunctive normal form of a Boolean function. The method of expressing a Boolean network as a disjunctive normal form has long been used in gating networks and switching functions, in particular in mapping and simplifying Boolean network expressions \cite{Fab}. Using the disjunctive normal form, we give an explicit algorithm on the fixed points. Although the algorithm is not polynomial in the number of variables, it is effective if the support (see definition in section 3) of the polynomial function that defines the Boolean network is relatively small. The described method is also useful when one wants to select a Boolean network to model a system based on experimental evidence, such as the construction of a Boolean network based on prescribed sets of attractors and transients. It should be pointed out that the problem of constructing Boolean networks from prescribed attractor structures was considered in \cite{Pal1} based on the truth table. However, our method emphasizes the rule played by the polynomial functions, and therefore, tools in computational algebra can be integrated into the computations. 
\par 
\medskip 
\section{Some basic properties of Boolean functions}
\par
\medskip
Let 
\be
f=(f_1, f_2,\ldots, f_n): F_2^n\longrightarrow F_2^n
\ee
be a Boolean function, where 
\be
f_i: F_2^n\longrightarrow F_2,\;\;1\le i\le n,
\ee
are the coordinate functions of $f$. It is well-known that all $f_i$ are elements of the Boolean ring 
\be
F_2[x_1,\cdots, x_n]/(x_i^2-x_i, 1\le i\le n),
\ee
i.e., the quotient of the polynomial ring $F_2[x_1,\cdots,x_n]$ by the ideal generated by $x_i^2-x_i, 1\le i\le n$. To simplify our notation, we will denote this Boolean ring by $F_2[x_1,\cdots,x_n]$, and use the following notation: 
\be
\mathbf{x}=(x_1,\ldots, x_n), \quad F_2[\mathbf{x}]=F_2[x_1,\ldots, x_n],\\
\mathbf{x}^{\mathbf{a}} = x_1^{a_1}\cdots x_n^{a_n},\quad\mbox{if $\mathbf{a}=(a_1,\ldots, a_n)$}.
\ee   
We also denote the $k$th iteration of $f$ by $f^k$, i.e.,
\be
f^k=\underbrace{f\circ f\circ\cdots\circ f}_{\mbox{$k$ terms}}, \quad k=1,2,\ldots.
\ee
\par
Note that all monomials of $F_2[\mathbf{x}]$ are square free monomials, i.e., a monomial can only be the form $x_1^{a_1}\cdots x_n^{a_n}$ with $a_i=0, 1$. Thus any $f\in F_2[\mathbf{x}]$ can be expressed in the form
\be
f(\mathbf{x})=\sum_{\mathbf{a}\in F_2^n}c_{\mathbf{a}}\mathbf{x}^{\mathbf{a}},\quad\mbox{where $c_{\mathbf{a}}\in F_2$}.
\ee
\par
The Boolean ring $F_2[\mathbf{x}]$ is also a Boolean algebra with the disjunction (OR) $\vee$, the conjunction (AND) $\wedge$, and the negation (NOT) ${}'$ defined by \cite{Rud1, Stone}
\be
f\vee g = f+g+fg,\quad f\wedge g= fg,\quad f'=1+f,\quad\forall f,g\in F_2[\mathbf{x}]. 
\ee
\par
\par
It should also be noted that some literature write $x^a$ for $x'=1+x$ if $a=0$, but in this paper, $x^0=1$, unless otherwise stated.
\par\medskip
Since the state space $F_2^n$ is a finite set, there is a positive integer $m$ such that $f^{m}$ is stabilized, i.e., $f^{m+r}=f^m,\;\forall r\ge 0$. The set $Y=f^m(F_2^n)$ is termed the {\it stable manifold} of $f$. The following simple observations hold.
\begin{proposition}\label{P1}
The transients of $f$ are given by $F_2^n-Y$. The restriction $f|_Y$ of $f$ to the stable manifold $Y$ of $f$ is a permutation of the set $Y$. 
\end{proposition} 
\par
\medskip
We collect some properties of the Boolean ring $F_2[\mathbf{x}]$ below. These properties are either contained in the standard references about $F_2[\mathbf{x}]$ \cite{Rud1, Stone} or immediate consequences of the basic properties therein, some of them can also be fund in \cite{Bric1}. 
\par
\begin{proposition}\label{P2}
All ideals of $F_2[\mathbf{x}]$ are principal, i.e., $F_2[\mathbf{x}]$ is a PI ring. All non-constant polynomials of $F_2[\mathbf{x}]$ are zero divisors, and hence the only unit is $1$.
\end{proposition} 
\begin{proof} If $f_1, \ldots, f_m$ generate the ideal $I$, then $\vee_{i=1}^mf_i$ also generates $I$. If $f$ is a non-constant polynomial, then so is $1+f$, and $f(1+f)=0$.
\end{proof}
\par
\medskip
Recall that a partial relation $\le$ can be defined on the Boolean ring $F_2[\mathbf{x}]$ as follows: for $a,b\in F_2[\mathbf{x}]$,  $a\le b$ if and only if $ab=a$. Write $a<b$ if $a\le b$ but $a\ne b$. An element $a\in F_2[\mathbf{x}]$ is called an {\it atom} if there is no element $f\in F_2[\mathbf{x}]$ such that $0<f<a$. The following proposition is a standard result.
\begin{proposition}\label{atom}
\par
(i) Every element in $F_2[\mathbf{x}]$ can be expressed as a unique disjunction (up to rearrangement) of atoms and the disjunction of all atoms in $F_2[\mathbf{x}]$ is equal to $1$.
\par
(ii) The atoms of $F_2[\mathbf{x}]$ are in one-to-one correspondence with the elements of $F_2^n$. The correspondence is defined by
\be
\mathbf{c}\longrightarrow a_{\mathbf{c}}(\mathbf{x}), \quad \mathbf{c}\in F_2^n,
\ee
where $a_{\mathbf{c}}(\mathbf{x})$ is defined by $a_{\mathbf{c}}(\mathbf{x})=1$ if $\mathbf{x}=\mathbf{c}$; and $0$ otherwise.
\end{proposition}
\par
\medskip
The atom corresponds to an element $\mathbf{c}=(c_1,c_2,\ldots,c_n)\in F_2^n$ can be given by the formula
\be
a_{\mathbf{c}}=(x_1+c_1+1)(x_2+c_2+1)\cdots(x_n+c_n+1)=\prod_{i=1}^n(x_i+c_i+1).
\ee
Note that if $c_i=1$, then $x_i+c_i+1=x_i$; and if $c_i=0$, then $x_i+c_i+1=x_i'$. 
\par
For example, the atom corresponds to $\mathbf{c}=(1,0,1,0,1)\in F_2^5$ is the function 
\be
a_{\mathbf{c}}=x_1x_2'x_3x_4'x_5.
\ee
\par
We remark that if $\mathbf{a}\ne\mathbf{b}$ are two elements of $F_2^n$, then the product of the corresponding atoms is $0$:
\bea\label{atom1}
a_{\mathbf{a}}a_{\mathbf{b}}=0,\quad\mbox{if $\mathbf{a}\ne\mathbf{b}$}.
\eea
This property implies that 
\bea
a_{\mathbf{a}}\vee a_{\mathbf{b}}=a_{\mathbf{a}}+ a_{\mathbf{b}},\quad\mbox{if $\mathbf{a}\ne\mathbf{b}$}.
\eea
Thus, a disjunction of different atoms is equal to the sum of the same set of atoms, and we can rewrite the first statement of Proposition \ref{atom} as
\begin{proposition}\label{atom2} 
Every element in $F_2[\mathbf{x}]$ can be expressed as a unique sum (up to rearrangement) of atoms and the sum of all atoms in $F_2[\mathbf{x}]$ is equal to $1$. 
\end{proposition}  
\par
\medskip
There is a one-to-one correspondence between the set $F_2^n$ and the set 
\be
\{0,1,\ldots, 2^n-1\}
\ee
of integers given by sending each integer $i\in\{0,1,\ldots, 2^n-1\}$ to its binary representation $\mathbf{i}=(i_1,\ldots,i_n)$, and we will use these notation interchangeably. 
\par\medskip
If $f\in F_2[\mathbf{x}]$, we will write $f(i)$ for $f(i_1,\ldots, i_n)$. Let
\bea
a_{\mathbf{i}}=\prod_{j=1}^n(x_j+i_j+1),\;\; i\in\{0,1,\ldots, 2^n-1\},
\eea
and order the elements of $F_2^n$ according to the natural order of the set $\{0,1,\ldots, 2^n-1\}$. Then by Proposition \ref{atom2} we have the following:
\begin{proposition}
Every $0\ne f\in F_2[\mathbf{x}]$ can be expressed as
\bea\label{dnf}
f=f(0)a_{\mathbf{0}}+f(1)a_{\mathbf{1}}+\cdots+f(2^n-1)a_{\mathbf{2^n-1}}.
\eea
\end{proposition} 
Now we are ready to give the following definition.
\begin{definition}
The expression in (\ref{dnf}) is called the disjunctive normal form (d.n.f) for $f\ne 0$. For $f=0$, the disjunctive normal form is $x_1x_1'$.
\end{definition}
\par
\medskip
It is straight forward to convert a function $f\in F_2[\mathbf{x}]$ from its polynomial expression to its disjunctive normal form: for each monomial in the expression of $f$, one fills in the missing variable $x_i$ with $x_i+x_i'$ and then simplify.
\begin{example} Consider $f=x_1x_2x_5+x_1x_2x_3x_4\in F_2^5[\mathbf{x}]$. Then we have
\be
f &=& x_1x_2(x_3+x_3')(x_4+x_4')x_5 + x_1x_2x_3x_4(x_5+x_5')\\
{}&=& x_1x_2x_3x_4'x_5 + x_1x_2x_3'x_4x_5 + x_1x_2x_3'x_4'x_5 + x_1x_2x_3x_4x_5'\\
{}&=& a_{\mathbf{25}} + a_{\mathbf{27}} + a_{\mathbf{29}} + a_{\mathbf{30}}.
\ee
\end{example}
\par
\medskip
\section{Dynamics of Boolean networks}
\par
\medskip
In this section, we will explain how to use the disjunctive normal form to derive information on the fixed points and limit cycles of a Boolean network. Fix a Boolean function
\be
f=(f_1, f_2,\ldots, f_n): F_2^n\longrightarrow F_2^n,
\ee
and define the support $\mbox{supp}(f)$ of $f$ to be the set of atoms that show up in the disjunctive normal forms of the $f_i$'s, i.e. 
\bea
\qquad\mbox{supp}(f)=\{a_{\mathbf{i}}\;|\; \exists f_j\;\mbox{such that}\; a_{\mathbf{i}} \;\mbox{appear in the d.n.f of $f_j$}\}.
\eea
We order the elements of supp$(f)$ according to the order of their indexes, and assume that 
\bea
\qquad\mbox{supp}(f)=\{a_{\mathbf{i}_1},a_{\mathbf{i}_2},\ldots, a_{\mathbf{i}_s}\}.
\eea
\par
\medskip
Then we can express each of the coordinate functions of $f$ as a linear combination of the elements in supp$(f)$:
\bea\label{lc}
f_i = \sum_{j=1}^sb_{ji}a_{\mathbf{i}_j},\;\;1\le i\le n,
\eea
where $b_{ji}\in F_2$. Using matrix notation, we can rewrite (\ref{lc}) as
\bea
f=(a_{\mathbf{i}_1},a_{\mathbf{i}_2},\ldots, a_{\mathbf{i}_s})B,
\eea
where $B$ is an $s\times n$ matrix whose entries are the $b_{ji}$'s defined by (\ref{lc}). More precisely, the $i$th column of $B$ is $(b_{1i},b_{2i},\ldots,b_{si})^T,\;1\le i\le n$.
\par
\medskip
Note that under our assumption, the corresponding integers $\{i_1,i_2,\ldots,i_s\}$ of the indexes (recall that the boldface notation $\mathbf{i}_j$ is the binary representation of $i_j$) of the elements in supp$(f)$ satisfy the relation
\be
0\le i_1< i_2<\cdots< i_s\le 2^n-1.
\ee
\par
Each row of the matrix $B$ can be viewed as an element of $F_2^n$, and we denote the corresponding integer of the $j$th row of $B$ by $r_j,\;1\le j\le s$. Note that these integers $r_j$ are not necessary distinct. We now define the following sets of integers:
\begin{gather}\label{set}
S = \{i_1,i_2,\ldots,i_s\}\cap\{r_1,r_2,\ldots,r_s\},\\
S_0 = \{i_j\;|\;i_j=r_j\}. \nonumber
\end{gather}
\par
\medskip
The following theorem describes the fixed points of $f$.
\begin{theorem}\label{T1} Notation as before. The fixed point(s) of $f$ can be described as follows.
\par
(i) If $0\notin\{i_1,i_2,\ldots,i_s\}$, then $0$ is a fixed point, and in this case, the set of fixed point(s) of $f$ is $S_0\cup\{0\}$.
\par
(ii) If $0\in\{i_1,i_2,\ldots,i_s\}$, then the set of fixed point(s) of $f$ is $S_0$.
\par
(iii) If $S=S_0$, then either $f$ is a fixed point system, or it has only one cycle of length $>1$, which is necessary length $2$ with one of the vertices being $0$.
\end{theorem}
\begin{proof} Using (\ref{lc}) we see that if $\mathbf{c}\in F_2^n$, then
\be
f_i(\mathbf{c})= \left\{\begin{array}{ccl}
                          0 &\mbox{if} & \mathbf{c}\notin\{\mathbf{i}_1,\mathbf{i}_2,\ldots,\mathbf{i}_s\};\\
                          b_{ji} &\mbox{if} & \mathbf{c}=\mathbf{i}_j\;\mbox{for some $1\le j\le s$}.
                         \end{array}\right.
\ee
Thus
\bea\label{fp}
\qquad f(\mathbf{c})= \left\{\begin{array}{lcl}
                          (0,0,\ldots,0) &\mbox{if} & \mathbf{c}\notin\{\mathbf{i}_1,\mathbf{i}_2,\ldots,\mathbf{i}_s\};\\
                          (b_{j1},b_{j2},\ldots,b_{jn}) &\mbox{if} & \mathbf{c}=\mathbf{i}_j\;\mbox{for some $1\le j\le s$}.
                         \end{array}\right.
\eea
If $0\notin\{i_1,i_2,\ldots,i_s\}$, then the first case in (\ref{fp}) implies that $\mathbf{c}\notin\{\mathbf{i}_1,\mathbf{i}_2,\ldots,\mathbf{i}_s\}$, is a fixed point if and only if $\mathbf{c}=0$. The second case in (\ref{fp}) implies that $\mathbf{i}_j, 1\le j\le s$ is a fixed point if and only if $\mathbf{i}_j=(b_{j1},b_{j2},\ldots,b_{jn})$, i.e. $i_j=r_j$. Therefore (i) follows.
\par
For (ii), we just need to notice that if $0\in\{i_1,i_2,\ldots,i_s\}$, then by the definition of supp$(f)$, $f(0)\ne 0$, and hence $0$ cannot be a fixed point.
\par
Now we prove (iii). Suppose that $S=S_0$ but $f$ is not a fixed point system. Let
\be
c_1\longrightarrow c_2\longrightarrow\cdots\longrightarrow c_k=c_0\longrightarrow c_1
\ee
be a cycle of length $k>1$. Then there are two possible cases: all $c_j\in\{i_1,i_2,\ldots,i_s\}$; or $c_j=0$ for some $1\le j\le k$ and $c_{j-1}\notin \{i_1,i_2,\ldots,i_s\}$. In the first case, all $c_j$ are also elements of $\{r_1,r_2,\ldots,r_s\}$ and hence belong to $S$. However, they are not fixed points and hence are not elements of $S_0$, which contradicts the assumption $S=S_0$. In the second case, if $k>2$, then $c_{j+1}\in S-S_0$, which is also a contradiction. So the only possibility is $k=2$, which can indeed happen. For instance, the function $f=x'\; :\; F_2\longrightarrow F_2$ has $\mbox{supp}(f)=\{0\}$, and since $f(0)=1$, $S=S_0=\emptyset$. But it has the cycle $\{0,1\}$.
\end{proof}
\par
\medskip
The statements about the fixed points of $f$ in Theorem \ref{T1} are explicit and can be readily turned into an algorithm. This will be discussed in the next section. We now consider the situation when $S\ne S_0$. In this case, $f$ may or may not be a fixed point system. We call a subset 
\be
V\subseteq \{i_1,i_2,\ldots,i_s\}\cup\{r_1,r_2,\ldots,r_s\}
\ee
an $f$-invariant set if $f(V)=V$. It is clear that $S_0$ is $f$-invariant. We denote by $M$ the maximum $f$-invariant subset. Then $S_0\subseteq M$. If $S_0\subsetneq M$, then $f|_{M-S_0}:=\sigma$ defines a permutation of $M-S_0$. Let
\bea\label{cy}
\sigma=\sigma_1\sigma_2\cdots\sigma_{h}
\eea
be the disjoint cycle decomposition of $\sigma$. Then we have
\begin{theorem} The Boolean network $f$ is a fixed point system if and only if $M=S_0$. If $M\ne S_0$, the limit cycles of $f$ of length $>1$ are given by the $\sigma_i,\;1\le i\le h$, defined by (\ref{cy}).
\end{theorem}
\begin{proof} This is straightforward. Note that we need to consider the set 
\be
\{i_1,i_2,\ldots,i_s\}\cup\{r_1,r_2,\ldots,r_s\},
\ee
since if $0\in \{i_1,i_2,\ldots,i_s\}$, then it can happen that there is a $c\notin \{i_1,i_2,\ldots,i_s\}$ such that $f(c)=0$.
\end{proof}
\par
\medskip
\section{Algorithm and examples} 
\par
\medskip
We now give an algorithm on the fixed points of a Boolean network from a given polynomial representation based on Theorem \ref{T1}. To better organize the computation, instead using the $x_i'$'s, we introduce new variables $y_1, y_2,\ldots, y_n$. We call the algorithm  BNFP (Boolean Network Fixed Point) algorithm. Note that in the previous section, the $i_j$'s correspond to the columns and the $r_j$'s correspond to the rows. However, in the computation, there is no need to distinguish between row vectors and column vectors. If we use row vector for the $i_j$'s, then the columns give the $r_j$'s. Also, in computer implementation, one may want to just use the binary numbers instead of their integer values.
\begin{algorithm} (BNFP)
\par
\medskip
{\bf INPUT:} $f=(f_1,f_2,\ldots,f_n),\; f_i\in F_2[x_1,x_2,\ldots,x_n],\; 1\le i\le n$.
\par
{\bf OUTPUT:} Fixed point(s) of $f$.
\par
\medskip
1. For each $1\le i\le n$ and each monomial term in $f_i$, if the $j$-th variable $x_j$ is missing, insert $x_j+y_j$ into the $j$-th place, simplify mod $2$.
\par
2. Assign an integer value to each of the terms $z_1z_2\cdots z_n,\; z_j\in\{x_j,y_j\}$, obtained in step 1 by first replacing $x_j$ by $1$ and replacing $y_j$ by $0$ to obtain the corresponding binary value. Assume that all the integer values thus obtained are
\be
i_1<i_2<\cdots<i_s.
\ee 
\par
3. For each $1\le i\le n$ and each expression $f_i$ obtained in step 1, assign a vector $c_i=(c_{i1},c_{i2},\ldots,c_{is})\in F_2^s$ by putting $1$ or $0$ as the $j$-th coordinate depending on whether $i_j$ shows up in the expression of $f_i$ or not.
\par
4. For each $1\le j\le s$, let $r_j$ be the integer value of the binary number $c_{1j}c_{2j}\cdots c_{nj}$.  
\par
5. Form $S_0=\{i_j\;|\; i_j=r_j\}$.
\par
6. If $0\in \{i_1,i_2,\ldots,i_s\}$, output $S_0$; otherwise, output $S_0\cup\{0\}$. 
\end{algorithm}
\par
\medskip
Let us use an example to explain the algorithm.
\par
\medskip
\begin{example} Consider the system $f: F_2^3\longrightarrow F_2^3$ defined by 
\be
f_1=x_1x_2+x_1x_2x_3,\;\; f_2=x_1x_3,\;\;f_3=x_1x_2x_3.
\ee
We walk through each of the steps in the algorithm:
\par
1. Insert terms and simplify:
\begin{gather*}
f_1= x_1x_2(x_3+y_3)+x_1x_2x_3 = x_1x_2y_3,\\
f_2=x_1(x_2+y_2)x_3 = x_1x_2x_3+x_1y_2x_3,\\
f_3=x_1x_2x_3.
\end{gather*} 
\par
2. Compute the integers $i_j$: $f_1$ gives the binary number $110$, which is $6$; $f_2$ gives $111$ and $101$, which are $7$ and $5$; $f_3$ gives $7$. Thus we have $3$ distinct integers ($s=3$): $5,6,7$.
\par
3. Since only $6$ appears in $f_1$, we have $c_1=(0,1,0)$. Similarly, $c_2=(1,0,1)$ and $c_3=(0,0,1)$.
\par
4. Reading the columns of the matrix with rows $c_1,c_2$, and $c_3$, we have $r_1=010=2$, $r_2=100=4$, and $r_3=011=3$.
\par
5. Thus $S_0=\emptyset$.
\par
6. Since $0\notin \{5,6,7\}$, the fixed point set is $\{0\}$.
\end{example}
\par
\medskip
We remark that for a monomial, there is no need to use the inserting computation to find out what integers we can get from it; and for a polynomial, we can find the integers each monomial term gives first, then gather them by noticing that two equal integers cancel each other. Here is an example.
\begin{example} Let $f= x_1x_2x_3x_5x_7+x_2x_3x_4x_6\in F_2[x_1,\ldots,x_7]$. The first monomial corresponds to $1110101\rightarrow 117$. It gives four integers, the other three can be obtained by adding $2$, $8$, and $10$, respectively, to $117$. They are $119$, $125$, and $127$. The numbers given by the second monomial are:
\be
58,\; 59,\; 62,\; 63,\; 122,\; 123,\; 126,\; 127.
\ee
Thus the support of $f$ is given by the integers: 
\be
58,\; 59,\; 62,\; 63,\; 117,\; 119,\; 122,\; 123,\; 125,\; 126.
\ee
\end{example}
\par
\medskip
We give an example to show how to use the described method to construct Boolean networks with given fixed points and cycle structure. 
\begin{example} We construct a network $f: F_2^6\longrightarrow F_2^6$ with the following fixed points and cycle structure: two fixed points $001011, 111100$; two cycles 
\begin{gather*}
010001\rightarrow 010011\rightarrow 011101 \rightarrow 011111\rightarrow 010001,\\
100100\rightarrow 101100 \rightarrow 100100. 
\end{gather*}
Since $0$ is not a fixed point, $0$ must appear in the support of $f$, we can take care of this by sending $0$ to, say, $111100$. The following functions, which are given in their disjunctive normal forms, give the coordinate functions of a Boolean network with the desired property:
\be
f_1&=& a_{\mathbf{0}}+a_{\mathbf{36}}+a_{\mathbf{44}}+a_{\mathbf{60}},\\
f_2 &=& a_{\mathbf{0}}+a_{\mathbf{17}}+a_{\mathbf{19}}+a_{\mathbf{29}}+a_{\mathbf{31}}+a_{\mathbf{60}},\\
f_3 &=& a_{\mathbf{0}}+a_{\mathbf{11}}+a_{\mathbf{19}}+a_{\mathbf{29}}+a_{\mathbf{36}}+a_{\mathbf{60}},\\
f_4 &=& a_{\mathbf{0}}+a_{\mathbf{19}}+a_{\mathbf{29}}+a_{\mathbf{36}}++a_{\mathbf{44}}+a_{\mathbf{60}},\\
f_5 &=& a_{\mathbf{11}}+a_{\mathbf{17}}+a_{\mathbf{29}},\\
f_6 &=& a_{\mathbf{11}}+a_{\mathbf{17}}+a_{\mathbf{19}}+a_{\mathbf{29}}+a_{\mathbf{31}}.\\
\ee
These functions are obtained by first converting the binary numbers that define the fixed points and the cycles to integers and then use the algorithm. The computation can be done by hand easily.
\end{example}
\par
\medskip
Note that though the conversion of these coordinate functions to polynomials in $x_1,\ldots,x_6$ is straightforward, the expressions are not as neat. 
\par
\medskip
\section{An example}
In this section, we use a Boolean network formulated in \cite{Albe} to model the expression pattern of the segment polarity genes in the fruit fly {\it Drosophila melanogaster} as an example to illustrate how we can apply our method to modify Boolean networks according to experimental evidence. Patterning in the early {\it Drosophila melanogaster} embryo is controlled by a protein regulatory network. There are $7$ segment polarity genes considered in \cite{Albe} including {\it wingless (wg)}, and {\it patched (ptc)}, with the corresponding proteins Wingless (WG) and Patched (PTC) respectively. The full model contains $21$ parameters and was later used as an example for the computational algebra approach to the reverse engineering of gene regulatory networks developed in \cite{Laub}. To analyze the stable states of the model, \cite{Albe} used a simplified Boolean network taking into consideration of the biological information. Here, for illustration purpose, we use also the simplified Boolean network. Note that we make no claim of the correctness of our modified model. According to \cite{Toml}, experimental evidence is yet to emerge to verify the existing models.
\par\medskip
The Boolean network we want to consider is given by the updating functions in Table 1.
\par\medskip
{\small
\begin{table}[h]
\begin{tabular}{c|l}
Node & Boolean updating function\\ \hline
$wg_1$ & $wg_1^{t+1} = wg_1\wedge\neg wg_2\wedge\neg wg_4$\\
$wg_2$ & $wg_2^{t+1} = \neg wg_1\wedge wg_2\wedge\neg wg_3$\\
$wg_3$ & $wg_3^{t+1} = wg_1\vee wg_3$\\
$wg_4$ & $wg_4^{t+1} = wg_2\vee wg_4$\\
$PTC_1$ & $PTC_1^{t+1} = (\neg wg_2\wedge\neg wg_4)\vee(PTC_1\neg wg_1\wedge\neg wg_3)$;\\ 
$PTC_2$ & $PTC_2^{t+1} = (\neg wg_1\wedge\neg wg_3)\vee(PTC_2\neg wg_2\wedge\neg wg_4)$;\\
$PTC_3$ & $PTC_3^{t+1} = 1$;\\
$PTC_4$ & $PTC_4^{t+1} = 1$.\\
\hline                               
\end{tabular}
\caption{Boolean model from \cite{Albe}}
\end{table}
}
For detailed information of this Boolean network, we refer the reader to \cite{Albe}. Since $PTC_3$ and $PTC_4$ remain unchanged in the process and they do not appear in the other updating rules, for computation purpose, we can ignore them. We introduce the variables as follows: $x_i$ for $wg_i$ for $1\le i\le 4$, $x_5$ for $PTC_1$, and $x_6$ for $PTC_2$. Then the polynomial function representation of the above Boolean network is:
{\footnotesize
\be
f_1 &=& x_1(x_2+1)(x_4+1),\\
f_2 &=& (x_1+1)x_2(x_3+1),\\
f_3 &=& x_1+x_3+x_1x_3,\\
f_4 &=& x_2+x_4+x_2x_4,\\
f_5 &=& (x_2+1)(x_4+1)+x_5(x_1+1)(x_3+1)+x_5(x_1+1)(x_2+1)(x_3+1)(x_4+1),\\
f_6 &=& (x_1+1)(x_3+1)+x_6(x_2+1)(x_4+1)+x_6(x_1+1)(x_2+1)(x_3+1)(x_4+1).
\ee
}
\par
It is easy to find the fixed points of this Boolean network, since there are only $64$ possible states\footnote{A useful tool to analyze small size Boolean networks is the Discrete Visualizer of Dynamics software developed by the Applied Discrete Mathematics Group at Virginia Bioinformatics Institute available at: www.vbi.vt.edu/admg/tools/.}. There are $10$ fixed points (the number in the parenthesis indicates the corresponding component size):  
\[\begin{array}{ccccl}
000101\;(2),\; &000111\;(2),\; &001010\;(2),\; &001011\;(2),\; &001100\;(36),\\
000011\;(4),\; &010101\;(4),\; &010111\;(4),\; &101010\;(4),\; &101011\;(4).
\end{array}
\]
Among these $10$ fixed points, only the two ($000101$ and $000011$) in the first column  and the one ($001100$) that lies in the component of size $36$ have been experimentally observed. Taking into consideration of the biological information again, the number can be reduced to $6$ \cite{Albe}. 
\par
Now suppose we want to modify this Boolean network so it contains only the $3$ known fixed points, and we want to redefine the images of the other $7$ fixed points as follows (this is just a random choice): 
\[\begin{array}{c}
010101\rightarrow 010111\rightarrow 011101,\\
101011\rightarrow 101010\rightarrow 001010\rightarrow 001011\rightarrow 000111\rightarrow 000101.
\end{array}
\]
Basically, the component to which $000011$ belongs is left alone, the first row redefines the function so that the components to which $010101$ and $010111$ belong will be added to the component with $36$ points, and the second row gathers the other $5$ components and add them to the component to which $000101$ belongs. Then we can apply our method to derive a new set of functions which define a Boolean network with the $3$ given fixed points:
{\footnotesize
\be
f_1 &=& x_1(x_4+1)(x_3x_5+x_3x_5x_6+1)(x_2+1),\\
f_2 &=& x_2(x_3+1)(x_1+1),\\
f_3 &=& x_3x_5x_6+x_1x_2x_3x_5x_6+x_1x_3x_4x_5x_6+x_1x_3+x_1+x_3+x_2x_3x_5x_6+\\
 {} &{}&   x_2x_4x_5x_6+x_3x_4x_5x_6+x_1x_3x_5x_6+x_1x_2x_4x_5x_6,\\
f_4 &=& x_1x_2x_4+x_3x_5x_6+x_1x_4+x_1x_2x_3x_5x_6+x_1x_3x_4x_5x_6+x_2x_4+\\
 {} &{}&   x_2+x_4+x_2x_3x_5x_6+x_3x_4x_5x_6+x_1x_2x_3x_4+x_1x_3x_5x_6+x_2x_3x_4x_5x_6+\\
 {} &{}&   x_1x_3x_4+x_1x_2x_3x_4x_5x_6,\\
f_5 &=& 1+x_1x_4x_5+x_1x_2x_5+x_4x_5x_6+x_2x_4x_5+x_2x_3x_5+x_2x_5+x_3x_4x_5+\\
 {} &{}&  x_2x_4x_6+x_1x_2x_3x_4x_5+x_1x_3x_4x_5x_6+x_2x_4+x_1x_2x_3x_4x_6+x_4x_5+x_2+\\
 {} &{}&  x_4+x_2x_3x_4x_5+x_2x_3x_4x_6+x_2x_4x_5x_6+x_3x_4x_5x_6+x_1x_4x_5x_6+\\
 {} &{}&  x_1x_2x_4x_5+x_1x_2x_3x_5+x_1x_3x_4x_5+x_1x_2x_4x_6+x_2x_3x_4x_5x_6+\\
 {} &{}&  x_1x_2x_4x_5x_6+x_1x_2x_3x_4x_5x_6,\\
f_6 &=& 1+x_2x_3x_6+x_3x_4x_6+x_1x_2x_6+x_3x_5x_6+x_2x_3x_5+x_1x_4x_6+x_3x_4x_5+\\
 {} &{}&    x_1x_3x_6+x_1x_3x_5+x_1x_2x_3x_4x_5+x_1x_6+x_3x_5+x_1x_2x_3x_4x_6+x_1x_3+\\
 {} &{}&     x_3x_6+x_1+x_3+x_2x_3x_4x_5+x_2x_3x_4x_6+x_2x_3x_5x_6+x_3x_4x_5x_6+\\
 {} &{}&     x_1x_2x_3x_6+x_1x_3x_4x_6+x_1x_2x_3x_5+x_1x_3x_4x_5+x_1x_2x_4x_6+\\
 {} &{}&     x_2x_3x_4x_5x_6.
\ee
}
\par
For this new network, the component to which the fixed point $000011$ belongs remains size $4$, the component to which the fixed point $000101$ belongs is now size $20$, and the component to which the fixed point $001100$ belongs is now size $40$.
\par\medskip
\section{Concluding remark}
\par
We described a method of deriving the dynamics of a Boolean network given in the form of a polynomial function using the disjunctive normal forms of the coordinate functions. This method can be used to construct Boolean networks with prescribed attractors and transients. The change of a Boolean network from its polynomial presentation to its disjunctive normal form is a change of bases procedure, since both the set of monomials and the set of atoms are bases of the vector space $F_2[\mathbf{x}]$ over $F_2$. However, the matrix of interchanging these two bases is of size $2^n\times 2^n$. Our method takes advantage of the fact that many Boolean networks have relatively small support compare to the number $2^n$, and makes computations involving only the support of the network. Another method of deriving the dynamics of a Boolean network is to use the truth table \cite{Pal1} (enumeration of the state space $S(f)$ using a table), which gives all the information about the corresponding Boolean network. In the worst case, when the support of a Boolean network is the whole space, our method involves the computation of $n$ functions and all $2^n$ points of the entire space, which then is equivalent to working with the whole truth table. Since the problem of linking the dynamics to the structure of a Boolean function is NP-hard in general, for application purpose, developing new algorithms that are not necessary polynomial time in $n$ but effective for some special classes is desirable.

\end{document}